\documentclass[conference]{IEEEtran}
\ifCLASSINFOpdf
\else
\fi
\usepackage[english]{babel}
\usepackage[T1]{fontenc}

\usepackage[utf8x]{inputenc}
\usepackage{cite}
\usepackage{graphicx}
\usepackage{color}
\usepackage{amssymb}
\usepackage{amsthm}
\usepackage{amsxtra}
\usepackage{amsmath}
\usepackage{multirow}
\usepackage{epsfig}
\usepackage{comment}
\usepackage{subfigure}
\usepackage{psfrag}
\usepackage{url}
\usepackage{textcomp}
\usepackage{enumerate}
\usepackage[printonlyused]{acronym}
\usepackage{epstopdf}
\usepackage{balance}
\usepackage{algorithm}
\usepackage{algpseudocode}
\usepackage{nomencl}
\usepackage{units}
\usepackage{pstricks, pst-plot, pst-grad, pstricks-add, pst-node, pstricks-add}
\algnewcommand{\Inputs}[1]{%
  \State \textbf{Inputs:}
   \hspace*{\algorithmicindent}\parbox[t]{.8\linewidth}{\raggedright #1}
}
\algnewcommand{\Initialize}[1]{%
  \State \textbf{Initialization:}
  \Statex \hspace*{\algorithmicindent}\parbox[t]{.8\linewidth}{\raggedright #1}
}
\algnewcommand{\Output}[1]{%
  \State \textbf{Output:}
   \hspace*{\algorithmicindent}\parbox[t]{.8\linewidth}{\raggedright #1}
}
\algnewcommand{\Update}[1]{%
  \State \textbf{Update:}
   \hspace*{\algorithmicindent}\parbox[t]{.8\linewidth}{\raggedright #1}
}
\usepackage{pifont}
\usepackage{glossaries}

\newtheorem{theorem}{Theorem}

\newcommand{\Matrix}[1]{\mathbf{#1}}

\newrgbcolor{darkgreen}{0.2 0.6 0.2}
\newrgbcolor{darkred}{0.6 0.2 0.2}
\newrgbcolor{darkblue}{0.2 0.2 0.6}
\newrgbcolor{purple}{0.2 0.1 0.9}
\begin{document}
%
\title{Delay constrained Energy Optimization for Edge Cloud Offloading}
\makeatletter
\def\ps@IEEEtitlepagestyle{
  \def\@oddfoot{\mycopyrightnotice}
  \def\@evenfoot{}
}
\def\mycopyrightnotice{
  {\footnotesize
  \begin{minipage}{\textwidth}
  \centering
 Copyright~\copyright~2018 IEEE ICC Workshops. Personal use of this material is permitted. Permission from IEEE must be obtained for all other uses, in any current or future media, including reprinting/republishing this material for advertising or promotional purposes, creating new collective works, for resale or redistribution to servers or lists, or reuse of any copyrighted component of this work in other works.
  \end{minipage}
  }
}


%

\author{
\IEEEauthorblockN{Shreya Tayade\IEEEauthorrefmark{1},
Peter Rost\IEEEauthorrefmark{2},
Andreas Maeder\IEEEauthorrefmark{2} and 
Hans D. Schotten\IEEEauthorrefmark{1}}
\IEEEauthorblockA{\IEEEauthorrefmark{1}University of Kaiserslautern,
Institute for Wireless Communications and Navigation,
Kaiserslautern, Germany\\
Email: \{tayade, schotten\}@eit.uni-kl.de}
\IEEEauthorblockA{\IEEEauthorrefmark{2}Nokia Bell Labs,
Munich, Germany\\
Email: \{peter.m.rost, andreas.maeder\}@nokia-bell-labs.com}
}

\maketitle

\begin{abstract}
Resource limited user-devices may offload computation to a cloud server, in order to reduce power consumption and lower the execution time. However, to communicate to the cloud server over a wireless channel, additional energy is consumed for transmitting the data. Also a delay is introduced for offloading the data and receiving the response. Therefore, an optimal decision needs to be made that would reduce the energy consumption, while simultaneously satisfying the delay constraint. In this paper, we obtain an optimal closed form solution for these decision variables in a multi-user scenario. Furthermore, we optimally allocate the cloud server resources to the user devices, and evaluate the minimum delay that the system can provide, for a given bandwidth and number of user devices.
\end{abstract}


%
\IEEEpeerreviewmaketitle
\section{Introduction}
Edge cloud offloading is a promising technique that enables resource-limited user devices to execute computationally extensive tasks. Cloud offloading has been broadly studied recently \cite{Kumar, Kumar2013, mao2017mobile,Cui2013}. From the user device perspective, to optimally offload the computation, two necessary conditions have to be satisfied: a) Energy consumption of the device should be minimal, and b) the offloaded task should be processed within a given latency constraint. However, although energy can be saved by offloading the computation, an additional energy is consumed for transmitting the data to the cloud. Furthermore, an additional transmitting and receiving delay is introduced for offloading the computation to the cloud. As their exists a trade-off between processing locally and offloading, we evaluate the optimal offloading decision.  

Many optimal offloading strategies have been proposed to reduce the energy consumption of user devices \cite{ Zhang2013, XudongXiang2014, Zhang2016,You2017}. In \cite{Zhang2013, XudongXiang2014}, the energy efficiency of a user device is increased by dynamically scheduling data transmission and link selection, as per the channel condition. 
Also, delay constrained, energy minimizing offloading techniques have been proposed in \cite{Kao,DongHuang2012,Zhang2017}. In \cite{Kao,DongHuang2012}, the authors partition a single task, and offload the individual partitions to the distributed cloud servers, ensuring that the execution is completed within a given deadline. However, the work in \cite{Zhang2013, XudongXiang2014, Kao,DongHuang2012} does not consider the multi-user effects on the cloud server, while taking an offloading decision. 

Furthermore, offloading computation also implies an additional cost of communication and cloud resources. The offloading decision is highly influenced by the availability of these resources \cite{Offloading2017}. Resources like bandwidth, cloud server capacity should be sufficient to satisfy the system requirements of ultra-low latency, and serve computational needs of all the user devices. At the same time, these resources must be used efficiently, which motivates the trade\-off analysis between these resources and the imposed delay requirements.
%
\cite{Sardellitti2015, Sardellitti2014} deal with joint optimization of the communication and computational resources for cloud offloading. However, no analysis on the trade-off between these resources and the achieved delay performance was presented. 

In this paper, a delay constrained energy optimization algorithm to optimally offload the computation to a cloud-server is designed. A closed form solution is provided for an optimal offloading decision. Also, the cloud resources are optimally allocated among multiple users. Furthermore, the delay performance of the system is analyzed for a given bandwidth.
In Section~\ref{sec: system:model}, we describe the system model. The energy optimization problem and the closed form solution is presented in Section~\ref{sec: sum.energy}. Finally, the results and conclusion are discussed in Section~\ref{sec:results} and \ref{sec:conclusion}, respectively.
\section{System Model} \label{sec: system:model}
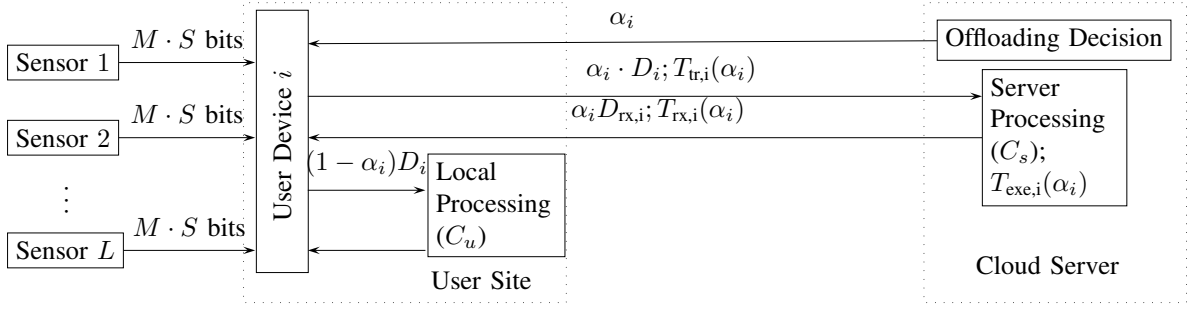
\begin{figure*}
  \centering
  \begingroup
\unitlength=1mm
\begin{picture}(155, 39)(0, 0)
\psset{xunit=1mm, yunit=1mm, linewidth=0.1mm}
\psset{arrowsize=2pt 3, arrowlength=1.6, arrowinset=.4}
 \rput(0, -3){
 	\rput[l](0, 35){\rnode{S1}{\psframebox{Sensor $1$}}}%
	\rput[l](0, 25){\rnode{S2}{\psframebox{Sensor $2$}}}%
    \rput[l](8, 17){\rput{90}{$\cdots$}}%
	\rput[l](0, 10){\rnode{SL}{\psframebox{Sensor $L$}}}%
    \rput[l](40, 5){\pnode(0, 32.5){MUXout}%
    	\pnode(-7, 30){MuxIn1}%
    	\pnode(-7, 20){MuxIn2}%
    	\pnode(-7, 5){MuxInL}%
        \pnode(0,25.5){Muxout2}
        \pnode(0,20){Muxout3}
        \pnode(0,13){Muxout4}
        \pnode(0,5){Muxout5}
        }
        
}
\rput[c]{90}(40.7,4){\psframe(35, 7){\rput[c](17, 3.5){User Device $i$}}}
\rnode{frame}{\rput[l](32, 0){\psframe[linestyle = dotted](43, 40){\rput[l](25,3){User Site}}}}
       \rput[l](55.3, 13){\psframebox{\parbox[c]{1.6cm}{Local Processing ($C_u$)}}}
\rput[l](120,0){
    \psframe[linestyle= dotted](35,40){\rput[l](7,5){Cloud Server}}
    }
	\rput[l](123,35){\rnode{offloading}{\psframebox{Offloading Decision}}}
    \rput[l](129,22){\rnode{servProcess}{\psframebox{\parbox[c]{1.7cm}{Server Processing ($C_s$); $T_\text{exe,i}(\alpha_i)$}}}}
   \pnode(129,27.5){servProcess3}
    \pnode(55,7){LocalProcess2}
    \pnode(55,15){LocalProcess}
    \ncline{->}{offloading}{MUXout}\nbput{$\alpha_i$}
    \ncline{->}{S1}{MuxIn1}\naput{$M\cdot S$ bits}
    \ncline{->}{S2}{MuxIn2}\naput{$M\cdot S$ bits}
    \ncline{->}{SL}{MuxInL}\naput{$M\cdot S$ bits}
    \ncline{->}{servProcess}{Muxout3}\nbput{$ \quad \alpha_i D_\text{rx,i}; T_\text{rx,i}(\alpha_i)$}
    \ncline{->}{Muxout2}{servProcess3}\naput{$ \quad \quad \alpha_i \cdot D_i; T_\text{tr,i}(\alpha_i)$}
    \ncline{->}{Muxout4}{LocalProcess}\naput{$(1-\alpha_i)D_i$}
    \ncline{->}{LocalProcess2}{Muxout5}
 
\end{picture}
\endgroup
  \caption{System model}
  \label{fig:system_model}
\end{figure*}
Consider $N$ uniformly distributed user devices in a circular area of radius $R$. An edge-cloud server is located at the base station in the center of the cell. The processor of an edge cloud server has a maximum computational capacity of $C_\text{s}$. The processors deployed in the user device have a maximum computational capacity of $C_\text{u}$, where, $C_\text{s} \gg C_\text{u}$. The user devices can successfully process all the data within the given time constraints. However, to minimize the energy consumption of user device, and reduce latency, user devices offload a share of data processing to the edge cloud server. The edge cloud server processes the data and sends the outcome of the computation to the user-device via the downlink channel as shown in Fig.~\ref{fig:system_model}. The uplink and downlink channel are known to the base-station for each user $i\in[1; N]$. 
\subsection{Data model}
\paragraph*{Uplink data model}
Every user device processes data from $L$ sensors. The data consists of $M$ data elements that are represented by $S$ bits each as shown in Fig.~\ref{fig:system_model}, e.\,g., surveillance by drones, where, $L$ cameras send images to the user device for processing. Each image is of pixel size $M$, where each pixel is represented by $S$ bits. Therefore, the total data bits that the user device needs to process is $D_i = L\cdot M\cdot S$. An algorithm to process this data has a complexity class given by the function $f_i(M)$. $f_i(M)$ represents the amount of computational cycles required, with respect to the number of data elements. The decision variable $\alpha_i$ denotes the share of data that should be offloaded, where $0\leq \alpha_i\leq 1$.  Hence, $\alpha_i \cdot D_i$ data bits are transmitted to an edge cloud server for processing. The processing algorithm is distributed to user device and edge cloud if 0 < $\alpha_i$ < 1.
\paragraph*{Downlink data model} Once the edge cloud server has completed all the data processing for the $i^\text{th}$ user device, it sends back the result to user device via the downlink channel. The total data bits sent from cloud server to the $i^\text{th}$ user device is $\alpha_i D_\text{rx,i}$, where $ D_\text{rx,i}= L \cdot S_\text{rx}$. $S_\text{rx}$ represents the number of bits used to encode the result of a single sensor of the device. 

\subsection{Delay model}\label{subsec: delay:model:111}
Let $T_\text{tr,i}$ be the total time for transmitting the data bits to the edge cloud server from the $i^\text{th}$ user. The available bandwidth $B$ is distributed equally among all the user devices, i.\,e., $B_i = B/N$ per user device. If $\alpha_i D_i$ are the total data bits transmitted over the channel with maximum spectral efficiency $R_i$, the total transmission time is given as
\begin{eqnarray}\label{eq:delay:optimization:100}
T_\text{tr,i}(\alpha_i) = \frac{\alpha_i D_i}{B_i R_i}. 
\end{eqnarray}
After transmission, the sensor data is processed in the edge cloud server. 

The execution time $T_\text{exe,i}$ to process the data, depends upon the computational load $C_\text{serv,i}$ introduced by each user device on the cloud, and the available cloud server capacity $C_s$. The computation load at the cloud server from $i^\text{th}$ user device is given as $C_\text{serv,i} = L \cdot \eta_s \cdot f_i(M)$, where $\eta_s$ is the number of CPU cycles required to process a single data element, for an algorithm of complexity $f_i(\cdot)$. 
Therefore, the execution time is given as: 
\begin{align}\label{eq:delay:optimization:110}
T_\text{exe,i}(\alpha_i) = \frac{\alpha_i  C_\text{serv,i}}{\rho_i C_s},
\end{align}
where $\rho_i$ represents the percentage of cloud resource allocated to the $i^\text{th}$ user device, and it holds $\sum_i^N \rho_i\leq1$; $\forall i = 1 \dots N; \rho_i\geq0$.\\
We further define $T_\text{rx,i}$ to be the time required to receive the processed result from the cloud server to the $i$-th user device, i.\,e.,
\begin{equation} \label{eq:delay:optimization:120}
T_\text{rx,i}(\alpha_i) = \frac{ \alpha_i D_\text{rx,i}}{B_\text{rx,i} \cdot R_\text{rx,i}},
\end{equation}  
where, $B_\text{rx,i}$ and $R_\text{rx,i}$ are the allocated bandwidth and maximum spectral efficiency respectively, for the $i^\text{th}$ user device in the downlink.

In order to fulfill the latency requirements, the total delay experienced by the user for transmitting, processing and receiving should be less than the maximum delay $T_\text{max}$: 
\begin{align}
\forall i = 1 \dots N: T_\text{tr,i}(\alpha_i) + T_\text{exe,i}(\alpha_i) + T_\text{rx,i}(\alpha_i) \leq T_\text{max}. 
\end{align}
For the sake of simplicity and due to the limited space of this paper, queuing delay is not considered in the model but a pre-reservation of computation resources at the cloud-server is assumed (as it would apply in hard real-time operating systems).
\subsection{Device-centric energy consumption model} The energy consumption model for the user devices is based on the model presented in\cite{Offloading2017}.
\paragraph*{Energy consumption for local processing}
The total energy consumed by the $i^\text{th}$ user device to locally process $(1-\alpha_i) D_i$ bits, is given as 
\begin{equation}
E_\text{u,i}(\alpha_i) = (1-\alpha_i) \cdot \epsilon_{i} \cdot C_\text{u,i} \label{eq: optimization:1}
\end{equation}
where $\epsilon_i$ is the average amount of energy consumed by the user device for a single computation cycle, and $C_\text{u,i}$ is the computation load generated in terms of computation cycles on the user device \cite{Offloading2017}. The computational load is given as $C_\text{u,i} = L \cdot \eta_i f_{i}(M)$ where $L$ is the number of sensors, and $\eta_i$ is a processor specific proportionality constant. $\eta_i$ represents the number of computation cycles required to process a single data element ($M=1$) for an algorithm of complexity $f_i$.
\paragraph*{Energy consumption for offloading}
The energy consumed to transmit $\alpha_i D_i$ data bits to the cloud with spectral efficiency $R_i$ is given as 
\begin{eqnarray} \label{eq: delay:optimization:130}
E_\text{tr,i}(\alpha_i)  =  \frac{\left(2^{R_i} - 1\right)}{G} \cdot \left [\frac{d_i}{d_o}\right]^{\beta} \cdot N_0 B_i \cdot T_\text{tr,i}(\alpha_i)  
\end{eqnarray}
where $\beta$ is the path-loss exponent, $d_i$ is the distance between user device and base station, $d_o$ is the reference distance, $N_0$ is the noise power spectral density, and $G$ is attenuation constant for free-space path-loss. 
Using $T_\text{tr,i}$ in~\eqref{eq:delay:optimization:100}, we get
\begin{eqnarray} \label{eq:optimization:2}
E_\text{tr,i}(\alpha_i)  =  \frac{\left(2^{R_i} - 1\right)}{G} \cdot \left [\frac{d_i}{d_o}\right]^{\beta} \cdot N_0 B_i \cdot \frac{\alpha_i D_i}{B_i R_i}. 
\end{eqnarray}
The total energy consumption at the $i^{th}$ user device is 
\begin{equation}
 E_\text{sum,i}(\alpha_i) = E_\text{u,i}(\alpha_i) + E_\text{tr,i}(\alpha_i).
\end{equation}

\section{Offloading Optimization} \label{sec: sum.energy}
\subsection{Problem formulation} Our objective is the derivation of an optimal offloading strategy that minimizes the energy consumption of the user device, while simultaneously ensuring that the total delay is below the threshold $T_\text{max}$, i.\,e., 
\begin{align} \label{eq:opt.sum.energy.1}
\mathcal{A}' & = \text{arg} \min_{\mathcal{A}\in\mathbb{R}^N} \: \sum_{i=1}^N \: E_\text{sum,i}(\alpha_i) \\
\text{s.t} & \quad  \sum\limits_i^N  \frac{\alpha_i C_\text{serv,i}}{T_\text{exe,i}(\alpha_i)} \leq C_s \label{eq:opt.sum.energy.2} \\ 
\mathcal{A} & = \{\alpha_1, \alpha_2, \dots \alpha_N\} \\
0 \leq \alpha_i \leq 1 & \quad \forall i = 1 \dots N 
\end{align}
In the given optimization problem, the decision vector $\mathcal{A'}$ is evaluated such that the total energy consumption of the user device is minimized. As the objective is to reduce the energy consumption of user devices, the energy consumed at the cloud server is not considered in the optimization problem. The second constraint for optimization ensures that the total processing rate required by the user devices should be less than the total server capacity. It reflects that the sum of the allocated shares of cloud resources $\rho_i$ must not exceed $1$, i.\,e.\,  $\sum_i^N \rho_i \leq 1$. $\alpha_i$ is the offloading decision parameter for the $i^\text{th}$ user. If $\alpha_i$ = 0, the user device do not offload, while if $\alpha_i$ = 1, all data processing is performed by the cloud server.

In order to serve more user devices on the cloud, it is necessary to efficiently allocate the cloud computational resources, $\rho_i$. The idea is to allocate more cloud server resources to the user devices that experience larger communication ($T_\text{tr,i} +T_\text{rx,i}$) delay, so as to reduce their execution time. As mentioned in Section~\ref{subsec: delay:model:111}, the total delay for offloading should be less than the maximum delay threshold $T_\text{max}$. Therefore, the maximum execution time permitted by the user device is 
\begin{eqnarray}\label{eq:delay:optimization:130}
T_\text{exe,i}(\alpha_i)  \leq  T_\text{max} - \left(T_\text{tr,i}(\alpha_i)  + T_\text{rx,i}(\alpha_i)\right).
\end{eqnarray}
By substituting ~\eqref{eq:delay:optimization:130} in ~\eqref{eq:opt.sum.energy.2}, the solution to the optimization problem is evaluated. As the optimization problem is convex, the solution is obtained by applying Lagrangian's duality theorem and KKT conditions.

\subsection{Solution}
\paragraph{Optimal offloading decision}
\begin{theorem} \label{thm:offloading:1}
If the rate of increase in energy consumption for local processing is higher than for offloading, i.e,\, $\left[-E_\text{tr,i}^{'} - E_\text{u,i}^{'}\right] > 0$, the optimal offloading decision for $i^\text{th}$ user device is given by 
\begin{equation} \label{eq:delay:optimization:132}
\alpha_i =  \min\left[1, \frac{1}{\left[ \frac{D_i}{B_i R_i} + \frac{D_\text{rx,i}}{B_\text{rx,i}R_\text{rx,i}}\right]} \left(T_\text{max} -  \sqrt{  \frac{ \nu \gamma_i  T_\text{max}}{(-E_\text{tr,i}^{'} - E_\text{u,i}^{'})}} \right)^+\right],
\end{equation}
where, $ E_\text{tr,i}^{'}(\alpha_i) = \frac{\partial E_\text{tr,i}(\alpha_i)}{\partial \alpha_i}$ and $E_\text{u,i}^{'}(\alpha_i) = \frac{\partial E_\text{u,i}(\alpha_i)}{\partial \alpha_i}$, $\nu$ is the Lagrange parameter defining the threshold for admitting the user devices, and $\gamma_i$ is the ratio of computational load to the cloud server capacity, given as $\gamma_i = \frac{C_\text{serv,i}}{C_s}$. \\
If the computation load $C_\text{serv,i} \ll C_s$, $\gamma_i \to 0$ and $\nu = 0$, the optimal offloading decision becomes
\begin{align}
\alpha_i = \min\left[1, \left( \frac{T_\text{max}}{\left[ \frac{D_i}{B_i R_i} + \frac{D_\text{rx,i}}{B_\text{rx,i}R_\text{rx,i}}\right]} \right)^+\right]
\end{align}
\end{theorem}
\begin{proof}
The proof of the theorem is given in Appendix.
\end{proof}
\begin{theorem}
In case of an overloaded system, i.\,e., $\nu \neq 0$, as $\alpha_i \geq 0$, the lower bound on Lagrange parameter $\nu$ is   
\begin{align} \label{eq:delay:optimization:140}
 \min\limits_{i: \alpha_i>0} \mathcal{B} \leq \nu, 
\end{align}
where, $\mathcal{B}=\{\hat\nu_1, \hat\nu_2,\dots \hat\nu_N \}$, and
\begin{align}
\hat\nu_i = \left[ \left(T_\text{max} -  \left[ \frac{D_i}{B_i R_i} + \frac{D_\text{rx,i}}{B_\text{rx,i}R_\text{rx,i}}\right]\right)^2 \frac{(-E_\text{tr,i}^{'} - E_\text{u,i}^{'})}{\gamma_i  T_\text{max}} \right]^+
\end{align}
\end{theorem}
\begin{proof}
  The theorem follows from Theorem \ref{thm:offloading:1} and applying the condition $\alpha\leq 1$ in (\ref{eq:delay:optimization:132}).
\end{proof}

\paragraph{Optimal cloud resource allocation} The relative share of cloud server capacity allocated to user device $i$ is given as 
\begin{equation}
\rho_i = \frac{\alpha_i C_\text{serv,i}}{C_s \cdot T_\text{exe,i}(\alpha_i)}.
\end{equation}
Hence, $\rho_i$ and $T_\text{exe,i}$ are a function of the optimal offloading decision $\alpha_i$. If the communication delay is higher then the execution time should be lower and the assigned computational share $\rho_i$ should be higher. In addition, the computational share $\rho_i$ scales linearly with the required computational load $\alpha_i C_\text{serv,i}$.

\subsection{Performance metrics}
\paragraph{Offloading percentage}
The offloading percentage is the ratio of total offloaded data processing for all user devices to the total data processing of the system, and it is given by 
\begin{align}
\Lambda =  \frac{\sum\limits_{i=1}^N \alpha_i \cdot D_i}{ \sum\limits_{i=1}^N D_i}.
\end{align}

\paragraph{Sum energy}
The performance of the offloading strategy is evaluated by comparing the total optimized energy consumption for all $N$ user devices to the total energy consumption for local processing. The total optimized energy consumption, i.\,e.,  $E_\text{sum}(\mathcal{A}')$, is given as
\begin{equation}
E_\text{sum}(\mathcal{A}') = \sum\limits_{i=1}^N \: E_\text{sum,i}(\alpha_i)
\end{equation}
where $\alpha_i$ is evaluated according to Theorem~\ref{thm:offloading:1} and ~\eqref{eq:delay:optimization:132}.
The total energy consumption in the case that no user device offloads ($\forall i \in [1; \dots N]: \alpha_i = 0$) is given by 
\begin{equation}
E_\text{sum}(\underline{0}) = \sum\limits_i^N \: E_\text{u,i}(0).
\end{equation}
\paragraph{Cut-off delay $T_c$} The minimum threshold delay within which the system can process all the optimally offloaded data from $N$ user devices, in presence of bandwidth $B$. 
\begin{align*}
T_c(B,N) = \inf \left\lbrace T_\text{max}>0 \biggl| \frac{\partial\Lambda(T_\text{max}, B, N)}{\partial T_\text{max}} = 0 \right\rbrace 
\end{align*}
\begin{algorithm}
\caption{Optimal Cloud Offloading}
\label{optimization}
\begin{algorithmic}
\Initialize{$\nu = 0$; $\alpha_i$ = 1, $\forall \: i = 1 \ldots N$; $\mathcal{B}=\{\hat\nu_1, \hat\nu_2,\dots \hat\nu_N\}$}\\
\mbox{Check on energy consumption:}
\For{User device i = 1:N}
\If{$\left[-E_\text{tr,i}^{'} - E_\text{u,i}^{'}\right] > 0$}
\State $\alpha_i$ = 1 \Comment{Offload}
\Else \State $\alpha_i$ = 0 \Comment{Local processing}
\EndIf
\EndFor \\
\mbox{Check load on the cloud server:}\\
\While{$\left( \sum\limits_i^N  \frac{  \alpha_i \gamma_i}{ \left[ T_\text{max} - (T_\text{tr,i}(\alpha_i) + T_\text{rx,i}(\alpha_i)) \right]} - 1 \right)$ > 0}\\
\State $\nu$ = $\min\limits_{\forall i= 1\dots N; \alpha_i>0} \{\mathcal{B}|\mathcal{B} > 0\}$ \\
\State Drop the user device that have highest communication delay and saves least energy, i.\,e.\, $\alpha_i = 0$ \\
\Update{$\mathcal{B}$, s.t $\mathcal{B} = \mathcal{B} - \{\nu \} $}
\State Assign $\alpha_i$ with new value of $\nu$, $\forall i = 1\dots N$
\begin{align}
\alpha_i =  \min \left[1, \frac{ \left(T_\text{max} -  \sqrt{ \frac{ \nu \gamma_i  T_\text{max}}{(-E_\text{tr,i}^{'} - E_\text{u,i}^{'})}} \right)^+}{\left[ \frac{D_i}{B_i R_i} + \frac{D_\text{rx,i}}{B_\text{rx,i}R_\text{rx,i}}\right]}\right] \nonumber
\end{align}
\EndWhile
\Output{$\mathcal{A'}$ and $\nu$} 
\end{algorithmic}
\end{algorithm}

\begin{figure}[!t]
\centering
\begingroup
\unitlength=1mm
\psset{xunit=3.59116mm, yunit=1.00000mm, linewidth=0.1mm}
\psset{arrowsize=2pt 3, arrowlength=1.4, arrowinset=.4}\psset{axesstyle=frame}
\begin{pspicture}(-2.17692, -16.00000)(20.10000, 50.00000)
\rput(-0.55692, -5.00000){%
\psaxes[subticks=0, labels=all, xsubticks=1, ysubticks=1, Ox=2, Oy=0, Dx=2, Dy=5]{-}(2.00000, 0.00000)(2.00000, 0.00000)(20.10000, 50.00000)%
\multips(4.00000, 0.00000)(2.00000, 0.0){8}{\psline[linecolor=black, linestyle=dotted, linewidth=0.2mm](0, 0)(0, 50.00000)}
\multips(2.00000, 5.00000)(0, 5.00000){9}{\psline[linecolor=black, linestyle=dotted, linewidth=0.2mm](0, 0)(18.10000, 0)}
\rput[b](11.05000, -11.00000){$ T_{max} \left[\text{ms}\right]  $}
\rput[t]{90}(-1.62000, 25.00000){$\Lambda$}
\psclip{\psframe(2.00000, 0.00000)(20.10000, 50.00000)}
\psline[linecolor=blue, plotstyle=curve, linewidth=0.4mm, showpoints=true, linestyle=solid, linecolor=blue, dotstyle=o, dotscale=1.2 1.2, linewidth=0.4mm](1.00000, 0.00003)(2.00000, 0.00236)(3.00000, 0.01194)(4.00000, 0.03327)(5.00000, 0.06640)(6.00000, 0.11673)(7.00000, 0.19091)(8.00000, 0.29850)(9.00000, 0.45086)(10.00000, 0.65181)(11.00000, 0.90978)(12.00000, 1.21979)(13.00000, 1.54257)(14.00000, 1.86536)(15.00000, 2.18814)(16.00000, 2.51092)(17.00000, 2.83371)(18.00000, 3.15649)(19.00000, 3.47927)(20.00000, 3.80206)
\psline[linecolor=red, plotstyle=curve, linewidth=0.4mm, showpoints=true, linestyle=solid, linecolor=red, dotstyle=diamond, dotscale=1.2 1.2, linewidth=0.4mm](1.00000, 0.00974)(2.00000, 0.09920)(3.00000, 0.38101)(4.00000, 0.97668)(5.00000, 1.96683)(6.00000, 3.30330)(7.00000, 4.73390)(8.00000, 6.16449)(9.00000, 7.59508)(10.00000, 9.02244)(11.00000, 10.44684)(12.00000, 13.41000)(13.00000, 13.41000)(14.00000, 13.41000)(15.00000, 13.41000)(16.00000, 13.41000)(17.00000, 13.41000)(18.00000, 13.41000)(19.00000, 13.41000)(20.00000, 13.41000)
\rput[l](11.00000, 15.00000){$T_c$}
\psline[linecolor=darkgreen, plotstyle=curve, linewidth=0.4mm, showpoints=true, linestyle=solid, linecolor=darkgreen, dotstyle=triangle, dotscale=1.2 1.2, linewidth=0.4mm](1.00000, 0.08522)(2.00000, 0.84022)(3.00000, 2.84468)(4.00000, 5.90265)(5.00000, 9.24148)(6.00000, 12.58030)(7.00000, 15.90180)(8.00000, 20.51600)(9.00000, 20.51600)(10.00000, 20.51600)(11.00000, 20.51600)(12.00000, 20.51600)(13.00000, 20.51600)(14.00000, 20.51600)(15.00000, 20.51600)(16.00000, 20.51600)(17.00000, 20.51600)(18.00000, 20.51600)(19.00000, 20.51600)(20.00000, 20.51600)
\psline[linecolor=cyan, plotstyle=curve, linewidth=0.4mm, showpoints=true, linestyle=solid, linecolor=cyan, dotstyle=square*, dotscale=1.2 1.2, linewidth=0.4mm](1.00000, 0.42119)(2.00000, 3.29106)(3.00000, 8.58162)(4.00000, 14.43987)(5.00000, 20.27865)(6.00000, 27.18400)(7.00000, 27.18400)(8.00000, 27.18400)(9.00000, 27.18400)(10.00000, 27.18400)(11.00000, 27.18400)(12.00000, 27.18400)(13.00000, 27.18400)(14.00000, 27.18400)(15.00000, 27.18400)(16.00000, 27.18400)(17.00000, 27.18400)(18.00000, 27.18400)(19.00000, 27.18400)(20.00000, 27.18400)
\psline[linecolor=darkred, plotstyle=curve, linewidth=0.4mm, showpoints=true, linestyle=solid, linecolor=darkred, dotstyle=triangle*, dotscale=1.2 1.2, linewidth=0.4mm](1.00000, 1.20893)(2.00000, 7.62826)(3.00000, 16.50227)(4.00000, 25.39196)(5.00000, 33.47800)(6.00000, 33.47800)(7.00000, 33.47800)(8.00000, 33.47800)(9.00000, 33.47800)(10.00000, 33.47800)(11.00000, 33.47800)(12.00000, 33.47800)(13.00000, 33.47800)(14.00000, 33.47800)(15.00000, 33.47800)(16.00000, 33.47800)(17.00000, 33.47800)(18.00000, 33.47800)(19.00000, 33.47800)(20.00000, 33.47800)
\endpsclip
\psframe[linecolor=black, fillstyle=solid, fillcolor=white, shadowcolor=lightgray, shadowsize=1mm, shadow=true](8.96154, 37.00000)(18.70769, 50.00000)
\rput[l](11.46769, 47.00000){\footnotesize{$\text{20} \% \text{B} $}}
\psline[linecolor=blue, linestyle=solid, linewidth=0.3mm](9.51846, 47.00000)(10.63231, 47.00000)
\psline[linecolor=blue, linestyle=solid, linewidth=0.3mm](9.51846, 47.00000)(10.63231, 47.00000)
\psdots[linecolor=blue, linestyle=solid, linewidth=0.3mm, dotstyle=o, dotscale=1.2 1.2, linecolor=blue](10.07538, 47.00000)
\rput[l](11.46769, 43.50000){\footnotesize{$\text{40} \% \text{B} $}}
\psline[linecolor=red, linestyle=solid, linewidth=0.3mm](9.51846, 43.50000)(10.63231, 43.50000)
\psline[linecolor=red, linestyle=solid, linewidth=0.3mm](9.51846, 43.50000)(10.63231, 43.50000)
\psdots[linecolor=red, linestyle=solid, linewidth=0.3mm, dotstyle=diamond, dotscale=1.2 1.2, linecolor=red](10.07538, 43.50000)
\rput[l](11.46769, 40.00000){\footnotesize{$\text{60} \%  \text{B} $}}
\psline[linecolor=darkgreen, linestyle=solid, linewidth=0.3mm](9.51846, 40.00000)(10.63231, 40.00000)
\psline[linecolor=darkgreen, linestyle=solid, linewidth=0.3mm](9.51846, 40.00000)(10.63231, 40.00000)
\psdots[linecolor=darkgreen, linestyle=solid, linewidth=0.3mm, dotstyle=triangle, dotscale=1.2 1.2, linecolor=darkgreen](10.07538, 40.00000)
\rput[l](15.64462, 47.00000){\footnotesize{$\text{80} \% \text{B} $}}
\psline[linecolor=cyan, linestyle=solid, linewidth=0.3mm](13.69538, 47.00000)(14.80923, 47.00000)
\psline[linecolor=cyan, linestyle=solid, linewidth=0.3mm](13.69538, 47.00000)(14.80923, 47.00000)
\psdots[linecolor=cyan, linestyle=solid, linewidth=0.3mm, dotstyle=square*, dotscale=1.2 1.2, linecolor=cyan](14.25231, 47.00000)
\rput[l](15.64462, 43.50000){\footnotesize{$\text{100} \%  \text{B} $}}
\psline[linecolor=darkred, linestyle=solid, linewidth=0.3mm](13.69538, 43.50000)(14.80923, 43.50000)
\psline[linecolor=darkred, linestyle=solid, linewidth=0.3mm](13.69538, 43.50000)(14.80923, 43.50000)
\psdots[linecolor=darkred, linestyle=solid, linewidth=0.3mm, dotstyle=triangle*, dotscale=1.2 1.2, linecolor=darkred](14.25231, 43.50000)
}\end{pspicture}
\endgroup
 
\caption{Offloading percentage} 
\label{fig:Delay_constraint:1}
\end{figure}
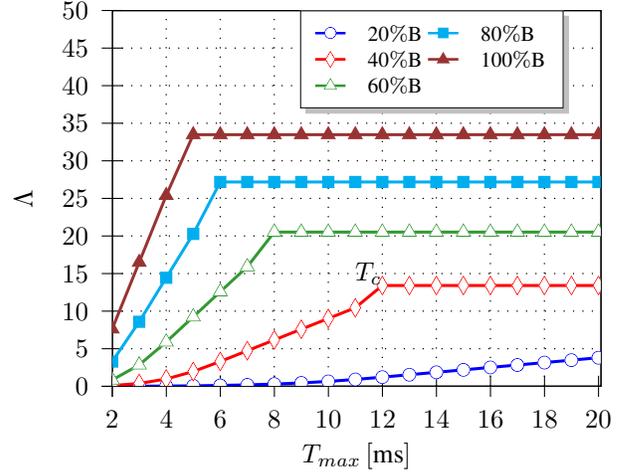

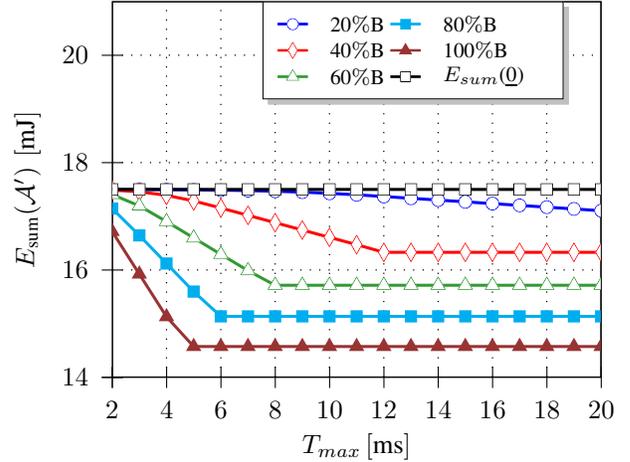
\begin{figure}[!t]
\centering
\begingroup
\unitlength=1mm
\psset{xunit=3.61111mm, yunit=7.14286mm, linewidth=0.1mm}
\psset{arrowsize=2pt 3, arrowlength=1.4, arrowinset=.4}\psset{axesstyle=frame}
\begin{pspicture}(-2.15385, 11.76000)(20.00000, 21.00000)
\rput(-0.55385, -0.70000){%
\psaxes[subticks=0, labels=all, xsubticks=1, ysubticks=1, Ox=2, Oy=14, Dx=2, Dy=2]{-}(2.00000, 14.00000)(2.00000, 14.00000)(20.00000, 21.00000)%
\multips(4.00000, 14.00000)(2.00000, 0.0){8}{\psline[linecolor=black, linestyle=dotted, linewidth=0.2mm](0, 0)(0, 7.00000)}
\multips(2.00000, 16.00000)(0, 2.00000){3}{\psline[linecolor=black, linestyle=dotted, linewidth=0.2mm](0, 0)(18.00000, 0)}
\rput[b](11.00000, 12.46000){$T_{max} \left[\text{ms}\right]   $}
\rput[t]{90}(-1.60000, 17.50000){$ E_\text{sum}(\mathcal{A'}) \: \left[\text{mJ}\right]$}
\psclip{\psframe(2.00000, 14.00000)(20.00000, 21.00000)}
\psline[linecolor=blue, plotstyle=curve, linewidth=0.4mm, showpoints=true, linestyle=solid, linecolor=blue, dotstyle=o, dotscale=1.2 1.2, linewidth=0.4mm](1.00000, 17.50000)(2.00000, 17.49970)(3.00000, 17.49852)(4.00000, 17.49604)(5.00000, 17.49224)(6.00000, 17.48658)(7.00000, 17.47830)(8.00000, 17.46625)(9.00000, 17.44930)(10.00000, 17.42740)(11.00000, 17.39999)(12.00000, 17.36803)(13.00000, 17.33497)(14.00000, 17.30192)(15.00000, 17.26886)(16.00000, 17.23580)(17.00000, 17.20274)(18.00000, 17.16968)(19.00000, 17.13662)(20.00000, 17.10357)
\psline[linecolor=red, plotstyle=curve, linewidth=0.4mm, showpoints=true, linestyle=solid, linecolor=red, dotstyle=diamond, dotscale=1.2 1.2, linewidth=0.4mm](1.00000, 17.49887)(2.00000, 17.48835)(3.00000, 17.45592)(4.00000, 17.39023)(5.00000, 17.28572)(6.00000, 17.15441)(7.00000, 17.01765)(8.00000, 16.88089)(9.00000, 16.74413)(10.00000, 16.60755)(11.00000, 16.47119)(12.00000, 16.32862)(13.00000, 16.32862)(14.00000, 16.32862)(15.00000, 16.32862)(16.00000, 16.32862)(17.00000, 16.32862)(18.00000, 16.32862)(19.00000, 16.32862)(20.00000, 16.32862)
\psline[linecolor=darkgreen, plotstyle=curve, linewidth=0.4mm, showpoints=true, linestyle=solid, linecolor=darkgreen, dotstyle=triangle, dotscale=1.2 1.2, linewidth=0.4mm](1.00000, 17.48990)(2.00000, 17.40578)(3.00000, 17.19360)(4.00000, 16.90208)(5.00000, 16.59682)(6.00000, 16.29156)(7.00000, 15.98707)(8.00000, 15.71498)(9.00000, 15.71498)(10.00000, 15.71498)(11.00000, 15.71498)(12.00000, 15.71498)(13.00000, 15.71498)(14.00000, 15.71498)(15.00000, 15.71498)(16.00000, 15.71498)(17.00000, 15.71498)(18.00000, 15.71498)(19.00000, 15.71498)(20.00000, 15.71498)
\psline[linecolor=cyan, plotstyle=curve, linewidth=0.4mm, showpoints=true, linestyle=solid, linecolor=cyan, dotstyle=square*, dotscale=1.2 1.2, linewidth=0.4mm](1.00000, 17.45146)(2.00000, 17.14399)(3.00000, 16.64415)(4.00000, 16.12033)(5.00000, 15.59728)(6.00000, 15.13567)(7.00000, 15.13567)(8.00000, 15.13567)(9.00000, 15.13567)(10.00000, 15.13567)(11.00000, 15.13567)(12.00000, 15.13567)(13.00000, 15.13567)(14.00000, 15.13567)(15.00000, 15.13567)(16.00000, 15.13567)(17.00000, 15.13567)(18.00000, 15.13567)(19.00000, 15.13567)(20.00000, 15.13567)
\psline[linecolor=darkred, plotstyle=curve, linewidth=0.4mm, showpoints=true, linestyle=solid, linecolor=darkred, dotstyle=triangle*, dotscale=1.2 1.2, linewidth=0.4mm](1.00000, 17.36344)(2.00000, 16.71632)(3.00000, 15.92241)(4.00000, 15.12874)(5.00000, 14.57401)(6.00000, 14.57401)(7.00000, 14.57401)(8.00000, 14.57401)(9.00000, 14.57401)(10.00000, 14.57401)(11.00000, 14.57401)(12.00000, 14.57401)(13.00000, 14.57401)(14.00000, 14.57401)(15.00000, 14.57401)(16.00000, 14.57401)(17.00000, 14.57401)(18.00000, 14.57401)(19.00000, 14.57401)(20.00000, 14.57401)
\psline[linecolor=black, plotstyle=curve, linewidth=0.4mm, showpoints=true, linestyle=solid, linecolor=black, dotstyle=square, dotscale=1.2 1.2, linewidth=0.4mm](1.00000, 17.50000)(2.00000, 17.50000)(3.00000, 17.50000)(4.00000, 17.50000)(5.00000, 17.50000)(6.00000, 17.50000)(7.00000, 17.50000)(8.00000, 17.50000)(9.00000, 17.50000)(10.00000, 17.50000)(11.00000, 17.50000)(12.00000, 17.50000)(13.00000, 17.50000)(14.00000, 17.50000)(15.00000, 17.50000)(16.00000, 17.50000)(17.00000, 17.50000)(18.00000, 17.50000)(19.00000, 17.50000)(20.00000, 17.50000)
\endpsclip
\psframe[linecolor=black, fillstyle=solid, fillcolor=white, shadowcolor=lightgray, shadowsize=1mm, shadow=true](7.53846, 19.18000)(18.61538, 21.00000)
\rput[l](10.03077, 20.58000){\footnotesize{$\text{20} \% \text{B} $}}
\psline[linecolor=blue, linestyle=solid, linewidth=0.3mm](8.09231, 20.58000)(9.20000, 20.58000)
\psline[linecolor=blue, linestyle=solid, linewidth=0.3mm](8.09231, 20.58000)(9.20000, 20.58000)
\psdots[linecolor=blue, linestyle=solid, linewidth=0.3mm, dotstyle=o, dotscale=1.2 1.2, linecolor=blue](8.64615, 20.58000)
\rput[l](10.03077, 20.09000){\footnotesize{$\text{40} \% \text{B} $}}
\psline[linecolor=red, linestyle=solid, linewidth=0.3mm](8.09231, 20.09000)(9.20000, 20.09000)
\psline[linecolor=red, linestyle=solid, linewidth=0.3mm](8.09231, 20.09000)(9.20000, 20.09000)
\psdots[linecolor=red, linestyle=solid, linewidth=0.3mm, dotstyle=diamond, dotscale=1.2 1.2, linecolor=red](8.64615, 20.09000)
\rput[l](10.03077, 19.60000){\footnotesize{$\text{60} \% \text{B} $}}
\psline[linecolor=darkgreen, linestyle=solid, linewidth=0.3mm](8.09231, 19.60000)(9.20000, 19.60000)
\psline[linecolor=darkgreen, linestyle=solid, linewidth=0.3mm](8.09231, 19.60000)(9.20000, 19.60000)
\psdots[linecolor=darkgreen, linestyle=solid, linewidth=0.3mm, dotstyle=triangle, dotscale=1.2 1.2, linecolor=darkgreen](8.64615, 19.60000)
\rput[l](14.18462, 20.58000){\footnotesize{$\text{80} \% \text{B} $}}
\psline[linecolor=cyan, linestyle=solid, linewidth=0.3mm](12.24615, 20.58000)(13.35385, 20.58000)
\psline[linecolor=cyan, linestyle=solid, linewidth=0.3mm](12.24615, 20.58000)(13.35385, 20.58000)
\psdots[linecolor=cyan, linestyle=solid, linewidth=0.3mm, dotstyle=square*, dotscale=1.2 1.2, linecolor=cyan](12.80000, 20.58000)
\rput[l](14.18462, 20.09000){\footnotesize{$\text{100} \% \text{B} $}}
\psline[linecolor=darkred, linestyle=solid, linewidth=0.3mm](12.24615, 20.09000)(13.35385, 20.09000)
\psline[linecolor=darkred, linestyle=solid, linewidth=0.3mm](12.24615, 20.09000)(13.35385, 20.09000)
\psdots[linecolor=darkred, linestyle=solid, linewidth=0.3mm, dotstyle=triangle*, dotscale=1.2 1.2, linecolor=darkred](12.80000, 20.09000)
\rput[l](14.18462, 19.60000){\footnotesize{$\text{$E_{sum}$(\underline{0})} $}}
\psline[linecolor=black, linestyle=solid, linewidth=0.3mm](12.24615, 19.60000)(13.35385, 19.60000)
\psline[linecolor=black, linestyle=solid, linewidth=0.3mm](12.24615, 19.60000)(13.35385, 19.60000)
\psdots[linecolor=black, linestyle=solid, linewidth=0.3mm, dotstyle=square, dotscale=1.2 1.2, linecolor=black](12.80000, 19.60000)
}\end{pspicture}
\endgroup
 
\caption{Energy consumption} 
\label{fig:Delay_constraint:2}
\end{figure}
%
\begin{table}[!t]
  \caption{Simulation Parameters}
  \label{tb:simulation:parameters}
  \centering
  \begin{tabular}{|c||c||c||c|}
    \hline 
    Variable & Value & Variable & Value \\ \hline
    \hline 
    $C_{s}$ & 200 MHz & $N$ & 50 Users \\
    \hline
    L & 10 & $M$ & 70 data elements \\
    \hline
    $B$, $B_\text{rx}$ & 20 MHz & $S$,$S_\text{rx}$ & 8 bits\\
    \hline
        $d_0$ & 200 m & $R$ & 800 m \\
    \hline
       $\beta$ & 2 &  $R_i, R_\text{rx}$ & $6$ bps/Hz \\
    \hline
    $\epsilon_i$ & $5e{-6}$ mJ & $\eta_i$ & 100 cycles \\
    \hline
     $f_i(M)$ & $M$  & $\eta_s$ & 1 cycle \\
    \hline
  \end{tabular}
\end{table}

\section{Results and Discussions} \label{sec:results}
\paragraph{Optimal offloading strategy and energy consumption}
Fig.~\ref{fig:Delay_constraint:1} shows the effect on the  average offloading percentage $\Lambda$, with an increasing delay $T_\text{max}$, and for different bandwidth $B$ available. The simulation parameters used for this evaluation are shown in Table \ref{tb:simulation:parameters}. \\
To optimally offload the data from the user devices, two crucial goals need to be achieved: a) save energy of the user device, and b) finish the computation within $T_\text{max}$. Consider the case when $40\% $ of bandwidth $B$ is available. At lower values of threshold delay, i.\,e.\, $T_\text{max} < \unit{12}{ms}$, the cloud server cannot process all the computational data within the given time constraint. Therefore, user devices only offload the data partially, ensuring that the offloaded data is computed within the specified delay. Therefore, more data is offloaded as the delay threshold is increased until the cut-off delay $T_c$ is reached. A further increase in the delay threshold $T_\text{max} > \unit{12}{ms}$ do not impact the offloading percentage $\Lambda$, which remains constant and is limited below $15\%$. This indicates that the cloud server could process more data for $T_\text{max} > \unit{12}{ms}$ but due to the limited bandwidth, it is not energy-efficient to offload more data. Therefore, the user devices do not offload, even for the lenient time constraints.
The optimal offloading percentage increases, if the percentage of available bandwidth is increased. Higher the bandwidth available, lower will be the data transmission time, and hence, the end-end delay. Therefore, at higher bandwidth ($60\%B$ and $80\%B$), the cut-off delay $T_c$ is reduced.  

Similar behavior is seen in Fig.~\ref{fig:Delay_constraint:2}. The energy consumption decreases, with an increase in the delay threshold $T_\text{max}$. This implies that energy consumption can be further reduced, if the cloud server can process more computation within the threshold delay. However, more computation could not be offloaded due to limited cloud server capacity and bandwidth. Therefore, at lower threshold delay, i.\,e., $T_\text{max}<T_c$, more energy can be saved by increasing the cloud server capacity. An increase in cloud server capacity will reduce the processing time, and hence will allow user devices to offload more computation. Whereas, if the threshold delay is higher, i.\,.e., $T_\text{max}>T_c$, the cloud server can effortlessly process more computation. An increase in cloud server capacity will have no impact on the energy consumption. Therefore, the energy consumption can only be further reduced by increasing the bandwidth.  



\paragraph{Bandwidth-delay-user devices trade-off}
\begin{figure}[!t]
\centering
\begingroup
\unitlength=1mm
\psset{xunit=0.81250mm, yunit=0.83333mm, linewidth=0.1mm}
\psset{arrowsize=2pt 3, arrowlength=1.4, arrowinset=.4}\psset{axesstyle=frame}
\begin{pspicture}(1.53846, -19.20000)(100.00000, 60.00000)
\rput(-2.46154, -6.00000){%
\psaxes[subticks=0, labels=all, xsubticks=1, ysubticks=1, Ox=20, Oy=0, Dx=20, Dy=10]{-}(20.00000, 0.00000)(20.00000, 0.00000)(100.00000, 60.00000)%
\multips(40.00000, 0.00000)(20.00000, 0.0){3}{\psline[linecolor=black, linestyle=dotted, linewidth=0.2mm](0, 0)(0, 60.00000)}
\multips(20.00000, 10.00000)(0, 10.00000){5}{\psline[linecolor=black, linestyle=dotted, linewidth=0.2mm](0, 0)(80.00000, 0)}
\rput[b](60.00000, -13.20000){$ \text{ Percentage of bandwidth utilized} $}
\rput[t]{90}(4.00000, 30.00000){$T_c$  [ms]}
\psclip{\psframe(20.00000, 0.00000)(100.00000, 60.00000)}
\psline[linecolor=blue, plotstyle=curve, linewidth=0.4mm, showpoints=true, linestyle=solid, linecolor=blue, dotstyle=o, dotscale=1.2 1.2, linewidth=0.4mm](20.00000, 10.00000)(40.00000, 5.00000)(60.00000, 4.00000)(80.00000, 3.00000)(100.00000, 2.00000)
\psline[linecolor=red, plotstyle=curve, linewidth=0.4mm, showpoints=true, linestyle=solid, linecolor=red, dotstyle=diamond, dotscale=1.2 1.2, linewidth=0.4mm](20.00000, 19.00000)(40.00000, 10.00000)(60.00000, 7.00000)(80.00000, 5.00000)(100.00000, 4.00000)
\psline[linecolor=darkgreen, plotstyle=curve, linewidth=0.4mm, showpoints=true, linestyle=solid, linecolor=darkgreen, dotstyle=triangle, dotscale=1.2 1.2, linewidth=0.4mm](20.00000, 29.00000)(40.00000, 15.00000)(60.00000, 10.00000)(80.00000, 8.00000)(100.00000, 6.00000)
\psline[linecolor=cyan, plotstyle=curve, linewidth=0.4mm, showpoints=true, linestyle=solid, linecolor=cyan, dotstyle=square*, dotscale=1.2 1.2, linewidth=0.4mm](20.00000, 38.00000)(40.00000, 19.00000)(60.00000, 13.00000)(80.00000, 10.00000)(100.00000, 8.00000)
\psline[linecolor=darkred, plotstyle=curve, linewidth=0.4mm, showpoints=true, linestyle=solid, linecolor=darkred, dotstyle=triangle*, dotscale=1.2 1.2, linewidth=0.4mm](20.00000, 47.00000)(40.00000, 24.00000)(60.00000, 16.00000)(80.00000, 12.00000)(100.00000, 10.00000)
\endpsclip
\psframe[linecolor=black, fillstyle=solid, fillcolor=white, shadowcolor=lightgray, shadowsize=1mm, shadow=true](63.07692, 36.00000)(93.84615, 60.00000)
\rput[l](74.15385, 56.40000){\footnotesize{$N = 20$}}
\psline[linecolor=blue, linestyle=solid, linewidth=0.3mm](65.53846, 56.40000)(70.46154, 56.40000)
\psline[linecolor=blue, linestyle=solid, linewidth=0.3mm](65.53846, 56.40000)(70.46154, 56.40000)
\psdots[linecolor=blue, linestyle=solid, linewidth=0.3mm, dotstyle=o, dotscale=1.2 1.2, linecolor=blue](68.00000, 56.40000)
\rput[l](74.15385, 52.20000){\footnotesize{$N = 40$}}
\psline[linecolor=red, linestyle=solid, linewidth=0.3mm](65.53846, 52.20000)(70.46154, 52.20000)
\psline[linecolor=red, linestyle=solid, linewidth=0.3mm](65.53846, 52.20000)(70.46154, 52.20000)
\psdots[linecolor=red, linestyle=solid, linewidth=0.3mm, dotstyle=diamond, dotscale=1.2 1.2, linecolor=red](68.00000, 52.20000)
\rput[l](74.15385, 48.00000){\footnotesize{$N = 60  $}}
\psline[linecolor=darkgreen, linestyle=solid, linewidth=0.3mm](65.53846, 48.00000)(70.46154, 48.00000)
\psline[linecolor=darkgreen, linestyle=solid, linewidth=0.3mm](65.53846, 48.00000)(70.46154, 48.00000)
\psdots[linecolor=darkgreen, linestyle=solid, linewidth=0.3mm, dotstyle=triangle, dotscale=1.2 1.2, linecolor=darkgreen](68.00000, 48.00000)
\rput[l](74.15385, 43.80000){\footnotesize{$N = 80 $}}
\psline[linecolor=cyan, linestyle=solid, linewidth=0.3mm](65.53846, 43.80000)(70.46154, 43.80000)
\psline[linecolor=cyan, linestyle=solid, linewidth=0.3mm](65.53846, 43.80000)(70.46154, 43.80000)
\psdots[linecolor=cyan, linestyle=solid, linewidth=0.3mm, dotstyle=square*, dotscale=1.2 1.2, linecolor=cyan](68.00000, 43.80000)
\rput[l](74.15385, 39.60000){\footnotesize{$N = 100 $}}
\psline[linecolor=darkred, linestyle=solid, linewidth=0.3mm](65.53846, 39.60000)(70.46154, 39.60000)
\psline[linecolor=darkred, linestyle=solid, linewidth=0.3mm](65.53846, 39.60000)(70.46154, 39.60000)
\psdots[linecolor=darkred, linestyle=solid, linewidth=0.3mm, dotstyle=triangle*, dotscale=1.2 1.2, linecolor=darkred](68.00000, 39.60000)
}\end{pspicture}
\endgroup
 
\caption{Bandwidth-delay trade-off} 
\label{fig:Delay_constraint:5}
\end{figure}
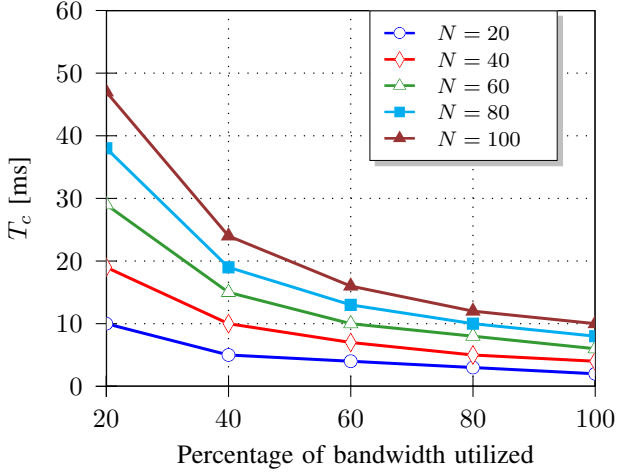
The results discussed previously show that there exists a trade-off between  bandwidth, $B$, number of users $N$ that can offload their computation, and cut-off delay $T_c$. If the number of user device increases, more bandwidth is required, and processing at the cloud server will take longer time. Whereas, if the bandwidth increases, the cut-off delay decreases, and more user devices can offload. Fig.~\ref{fig:Delay_constraint:5} demonstrates the behavior of this trade-off. The aim is to serve the maximum number of user devices with minimum bandwidth, at very low cut-off delay. 
Fig.~\ref{fig:Delay_constraint:5} illustrates the rate at which the cut-off delay decreases with respect to bandwidth. The cut-off delay decreases exponentially with an increase in bandwidth. However, when the number of user devices is lower, bandwidth has less impact on the cut-off delay. This is due to the fact that the time to transmit and receive is significantly lower compared to the execution time.


\section{Conclusion}\label{sec:conclusion}
In this paper, we presented a delay constrained and energy optimized cloud offloading framework. A closed form solution is obtained to optimally offload the computation from multiple user devices, depending on the availability of bandwidth and latency constraints. The optimal offloading strategy is to offload more computation to the cloud, if sufficient bandwidth and cloud server capacity is available. 
At higher threshold delay, the offloading decision is independent of the delay constraint. However, it depends upon the difference of the energy for transmitting and local processing. We also analyzed the trade-off between the bandwidth and cut-off delay considering the optimal offloading solution. 
In the future, we plan to extend the framework to study the effects of channel fading and shadowing on the optimal offloading strategy. 
\appendix 
In the following, we prove Theorem \ref{thm:offloading:1} using the method of Lagrange multipliers and Karush-Kuhn-Tucker (KKT) conditions.
Substituting the maximum value of $T_\text{exe}(\alpha_i)$ in ~\eqref{eq:opt.sum.energy.2}, and replacing $\left(\frac{C_\text{serv,i}}{C_s}\right) = \gamma_i$, the constraint~\eqref{eq:opt.sum.energy.2} of optimization problem becomes 
$\sum\limits_i^N \frac{\alpha_i \gamma_i}{\left[ T_\text{max} - T_\text{rx}(\alpha_i) - T_\text{tr}(\alpha_i) \right]} \leq 1$.
The objective function for the optimization is hence given as 
\begin{align}
\mathcal{L}(\alpha_i, \psi, \nu) = \sum_{i=1}^N \: E_\text{sum,i}(\alpha_i) + \nonumber \\ \nu \left( \sum\limits_i^N  \frac{ \alpha_i \gamma_i}{\left[ T_\text{max} - (T_\text{tr,i}(\alpha_i) + T_\text{rx,i}(\alpha_i)) \right]} - 1 \right) - \text{tr}\left[\Matrix{\Psi} \text{diag}(\alpha_i)\right]
\end{align}
where, $\nu$ and $\psi$ are the Lagrange's multiplier.
Take partial derivative with respect to $\alpha_i$ and equating $\frac{\partial \mathcal{L}}{\partial \alpha_i}$ to 0, we get,
\begin{align} \label{eq:delay:optimization:121}
\frac{\partial \mathcal{L}}{\partial \alpha_i} = E_\text{tr,i}^{'} + E_\text{u,i}^{'} + \nu \cdot \frac{ \gamma_i T_\text{max} }{\left[ T_\text{max} -  T_\text{rx,i}(\alpha_i) - T_\text{tr}(\alpha_i)  \right]^2} - \psi_i \nonumber \\
= 0 
\end{align}
\begin{align} 
\text{Where,} \quad E_\text{tr,i}^{'}(\alpha_i) = \frac{\partial E_\text{tr,i}(\alpha_i)}{\partial \alpha_i} \: \text{and} \: E_\text{u,i}^{'}(\alpha_i) = \frac{\partial E_\text{u,i}(\alpha_i)}{\partial \alpha_i}.
\end{align}
$E_\text{tr,i}^{'}(\alpha_i)$ and $E_\text{u,i}^{'}(\alpha_i)$ are obtained by taking derivative of ~\eqref{eq:optimization:2} and ~\eqref{eq: optimization:1} respectively.
\begin{align}
E_\text{tr,i}^{'} &= \frac{2^{R_i} -1}{G} \cdot \left [\frac{d_i}{d_o}\right]^{\beta} \cdot N_0 \cdot \frac{D_i}{R_i} \\ 
E_\text{u,i}^{'} &= - L \cdot \epsilon_i \eta_i f_i(M).
\end{align}

Convexity check: Take second derivative of $\mathcal{L}(\alpha_i,\nu,\psi_i)$
\begin{eqnarray}
\frac{\partial^2 \mathcal{L}}{\partial \alpha_i^2} = 2 \nu \gamma_i T_\text{max}(T_\text{max} - \left[ T_\text{tr,i}(\alpha_i) + T_\text{rx,i}(\alpha_i)\right])^{-3}  \label{eq:convexity:1} \\
\text{where,} \quad T_\text{max} > (T_\text{tr,i}(\alpha_i) + T_\text{rx,i}(\alpha_i)). \nonumber \\
\frac{\partial^2 \mathcal{L}}{\partial \alpha_i \alpha_j} = 0, \forall i \neq j. \label{eq:convexity:2}
\end{eqnarray}
Substituting the values of ~\eqref{eq:convexity:1} and ~\eqref{eq:convexity:2} in Hessian matrix, shows that it is positive semi-definite, hence, the objective function is a convex function.\\
KKT Conditions: 
\begin{eqnarray}
\nu \left( \sum\limits_i^N  \frac{  \alpha_i \gamma_i}{ \left[ T_\text{max} - (T_\text{tr,i}(\alpha_i) + T_\text{rx,i}(\alpha_i)) \right]} - 1 \right) = 0 \\
\nu \geq 0 ; \psi_i \alpha_i = 0;\\
\psi_i \geq 0; \alpha_i \geq 0;
\end{eqnarray}
CASE I: Fully-loaded case ($\nu > 0; \psi_i = 0$) \\ 
If $\nu$ > 0, i.e\, $\left( \sum\limits_i^N  \frac{ \alpha_i \gamma_i}{ \left[ T_\text{max} - (T_\text{tr,i}(\alpha_i) + T_\text{rx,i}(\alpha_i)) \right]} - 1 \right)$ = 0.
If $\psi_i$ = 0, i.e\, $\alpha_i$ > 0.
Substitute the values of $\psi_i = 0$ in ~\eqref{eq:delay:optimization:121}, and rearranging to evaluate $\alpha_i$, we get
\begin{align}
\alpha_i =  \frac{1}{\frac{D_i}{B_i R_i} + \frac{D_\text{rx,i}}{B_\text{rx,i} R_\text{rx,i}}} \left(T_\text{max}  -  \sqrt{  \frac{ \nu \gamma_i T_\text{max}}{(-E_\text{tr,i}^{'} - E_\text{u,i}^{'})}} \right) \label{eq:delay:optimization:110}
\end{align}
CASE II: Underloaded ($\nu$ = 0; $\psi_i$ = 0)\\
If $\nu$ = 0, implies that $\left( \sum\limits_i^N  \frac{  \alpha_i \gamma_i}{ \left[ T_\text{max} - (T_\text{tr,i}(\alpha_i) + T_\text{rx,i}(\alpha_i)) \right]} - 1 \right)$ <  0.\\ 
It states that total required processing rate is less than the server capacity. It means cloud server can process all the computation data (i.\,e.\, $\alpha_i = 1$) from all user devices within the specified time. Substitute the values $\nu$ = 0 in ~\eqref{eq:delay:optimization:121} we get,
\begin{equation} \label{eq:delay:optimization:160}
E_\text{tr,i}'+ E_\text{u,i}' = \psi_i 
\end{equation}
Now, put $\psi_i$ = 0 in ~\eqref{eq:delay:optimization:160}.
In this case user devices can offload to the cloud if the transmission energy is less than or equal to local processing energy. \\
CASE III: Overloaded ($\nu$ > 0, $\psi_i$ > 0 i.\,e.\, $\alpha_i$ = 0) \\
Substitute $\alpha_i = 0$ in ~\eqref{eq:delay:optimization:121}
\begin{align}
 E_\text{tr,i}^{'} + E_\text{u,i}^{'} + \frac{\nu \gamma_i T_\text{max}}{T_\text{max}} - \psi_i = 0   \nonumber \\
\psi_i =  E_\text{tr,i}^{'} + E_\text{u,i}^{'} + \nu \gamma_i
\end{align}
The Lagrange multiplier $\nu > 0$ and $\gamma_i > 0$. Therefore, for $\psi_i > 0$, i.\,e.\, $\alpha_i = 0$, the rate at which transmission energy $E_\text{tr,i}^{'}$ increases, has to be greater than the rate of increase in energy consumption due to local processing $E_\text{u,i}^{'}$. However, in this case as $\nu>0$, that means the cloud server is already overloaded. Hence, user device do not offload even if the transmission energy is less than the local processing energy consumption, unless the difference  between $E_\text{u,i}^{'}$ and $E_\text{tr,i}^{'}$  exceeds the value $\nu \gamma_i$.  \\
CASE IV: Underloaded ($\nu$ = 0, $\psi_i$ > 0; $\alpha_i$ = 0); 
If $\nu$ = 0, then $\left( \sum\limits_i^N  \frac{  \alpha_i \gamma_i}{ \left[ T_\text{max} - (T_\text{tr,i}(\alpha_i) + T_\text{rx,i}(\alpha_i)) \right]} - 1 \right) <  0;$
Now, substitute the values $\nu$ = 0 and $\alpha$ = 0 in equation~\eqref{eq:delay:optimization:121}, we get 
\begin{align}
  E_\text{tr,i}^{'} + E_\text{u,i}^{'} = \psi_i
\end{align}
As $\psi_i$ > 0, if $\alpha_i$ = 0, it means that if the transmission energy exceeds the in-device energy consumption, do not offload.\\
\balance
\bibliographystyle{IEEEtran}
\bibliography{ICC_paper_2017,ICC_v1}

\begin{thebibliography}{10}
\providecommand{\url}[1]{#1}
\csname url@samestyle\endcsname
\providecommand{\newblock}{\relax}
\providecommand{\bibinfo}[2]{#2}
\providecommand{\BIBentrySTDinterwordspacing}{\spaceskip=0pt\relax}
\providecommand{\BIBentryALTinterwordstretchfactor}{4}
\providecommand{\BIBentryALTinterwordspacing}{\spaceskip=\fontdimen2\font plus
\BIBentryALTinterwordstretchfactor\fontdimen3\font minus
  \fontdimen4\font\relax}
\providecommand{\BIBforeignlanguage}[2]{{%
\expandafter\ifx\csname l@#1\endcsname\relax
\typeout{** WARNING: IEEEtran.bst: No hyphenation pattern has been}%
\typeout{** loaded for the language `#1'. Using the pattern for}%
\typeout{** the default language instead.}%
\else
\language=\csname l@#1\endcsname
\fi
#2}}
\providecommand{\BIBdecl}{\relax}
\BIBdecl

\bibitem{Kumar}
K.~Kumar and Y.-H. Lu, ``{Cloud Computing for Mobile Users: Can Offloading
  Computation Save Energy? Techniques to save energy for mobile systems}.''

\bibitem{Kumar2013}
K.~Kumar, J.~Liu, Y.-H. Lu, and B.~Bhargava, ``{A Survey of Computation
  Offloading for Mobile Systems},'' \emph{Mob. Networks Appl.}, no.~1, pp.
  129--140, feb 2013.

\bibitem{mao2017mobile}
Y.~Mao, C.~You, J.~Zhang, K.~Huang, and K.~B. Letaief, ``Mobile edge computing:
  Survey and research outlook,'' \emph{arXiv preprint arXiv:1701.01090}, 2017.

\bibitem{Cui2013}
Y.~Cui, X.~Ma, H.~Wang, I.~Stojmenovic, J.~Liu, Y.~Cui, X.~Ma,
  {\textperiodcentered}.~H. Wang, I.~Stojmenovic, and J.~Liu, ``{A Survey of
  Energy Efficient Wireless Transmission and Modeling in Mobile Cloud
  Computing},'' \emph{Mob. Netw Appl}, vol.~18, pp. 148--155, 2013.

\bibitem{Zhang2013}
W.~Zhang, Y.~Wen, K.~Guan, D.~Kilper, H.~Luo, and D.~O. Wu, ``{Energy-Optimal
  Mobile Cloud Computing under Stochastic Wireless Channel},'' \emph{IEEE
  Trans. Wirel. Commun.}, vol.~12, no.~9, pp. 4569--4581, sep 2013.

\bibitem{XudongXiang2014}
{Xudong Xiang}, {Chuang Lin}, and {Xin Chen}, ``{Energy-Efficient Link
  Selection and Transmission Scheduling in Mobile Cloud Computing},''
  \emph{IEEE Wirel. Commun. Lett.}, vol.~3, no.~2, pp. 153--156, apr 2014.

\bibitem{Zhang2016}
K.~Zhang, Y.~Mao, S.~Leng, Q.~Zhao, L.~Li, X.~Peng, L.~Pan, S.~Maharjan, and
  Y.~Zhang, ``{Energy-Efficient Offloading for Mobile Edge Computing in 5G
  Heterogeneous Networks},'' \emph{IEEE Access}, vol.~4, pp. 5896--5907, 2016.

\bibitem{You2017}
C.~You, K.~Huang, H.~Chae, and B.-H. Kim, ``{Energy-Efficient Resource
  Allocation for Mobile-Edge Computation Offloading},'' \emph{IEEE Trans.
  Wirel. Commun.}, vol.~16, no.~3, pp. 1397--1411, mar 2017.

\bibitem{Kao}
Y.-H. Kao and B.~Krishnamachari, ``{Optimizing Mobile Computational Offloading
  with Delay Constraints}.''

\bibitem{DongHuang2012}
D.~Huang, P.~Wang, and D.~Niyato, ``A dynamic offloading algorithm for mobile
  computing,'' \emph{IEEE Transactions on Wireless Communications}, vol.~11,
  no.~6, pp. 1991--1995, June 2012.

\bibitem{Zhang2017}
K.~Zhang, Y.~Mao, S.~Leng, S.~Maharjan, and Y.~Zhang, ``{Optimal delay
  constrained offloading for vehicular edge computing networks},'' in
  \emph{2017 IEEE Int. Conf. Commun.}\hskip 1em plus 0.5em minus 0.4em\relax
  IEEE, may 2017, pp. 1--6.

\bibitem{Offloading2017}
S.~Tayade, P.~Rost, A.~Maeder, and H.~Schotten, ``{Device-centric Energy
  Optimization for Edge Cloud Offloading},'' in \emph{Globecom 2017}.\hskip 1em
  plus 0.5em minus 0.4em\relax IEEE, 2017.

\bibitem{Sardellitti2015}
S.~Sardellitti, G.~Scutari, and S.~Barbarossa, ``Joint optimization of radio
  and computational resources for multicell mobile-edge computing,'' \emph{IEEE
  Transactions on Signal and Information Processing over Networks}, vol.~1,
  no.~2, pp. 89--103, June 2015.

\bibitem{Sardellitti2014}
------, ``{Distributed joint optimization of radio and computational resources
  for mobile cloud computing},'' in \emph{2014 IEEE 3rd Int. Conf. Cloud
  Netw.}\hskip 1em plus 0.5em minus 0.4em\relax IEEE, oct 2014, pp. 211--216.

\end{thebibliography}

\end{document}